\documentclass{article}

\usepackage{etex}
\usepackage[margin=3.7cm]{geometry}
\usepackage{algcompatible}
\usepackage{algpseudocode,algorithm,algorithmicx}
\usepackage{fullwidth}
\usepackage{float}
\usepackage{tikz}
\usepackage{multicol}
\usepackage{todonotes}
\usepackage{wrapfig}
\usepackage{adjustbox}
\usepackage{tkz-graph}
\usetikzlibrary{arrows,automata}
\usetikzlibrary{arrows,shapes,positioning,scopes}
\usetikzlibrary{tikzmark,decorations.pathreplacing,calc,fit}
\usetikzlibrary{arrows,automata}
\usetikzlibrary{arrows,petri}
\GraphInit[vstyle = Shade]
\usepackage{ulem}
\usepackage[latin1]{inputenc}
\usepackage{ifthen}
\usepackage{mathrsfs}
\usepackage{multirow}
\usepackage{rotating}
\usetikzlibrary{matrix,arrows}
\usepackage{bussproofs}
\usepackage{qtree}
\usepackage{easyReview}
\usepackage[polutonikogreek,english]{babel}
\usepackage{tikz}
\usepackage{adjustbox}

\usetikzlibrary{arrows,petri}

\usepackage{pgf}
\usepackage{times}
\usepackage[T1]{fontenc}
\usepackage{amsfonts,amsmath,amssymb,amsthm}
\usepackage[all]{xy}

\usepackage{dingbat}
\usepackage{cmll}
\floatstyle{plain}
\floatstyle{ruled}
\restylefloat{algorithm}
\newfloat{myalgo}{tbhp}{mya}
\newcommand{\INDSTATE}[1][1]{\State\hspace{1cm}}
\newlength\myindent
\newtheorem{observation}{Observation}

\algrenewcommand\algorithmicrequire{\textbf{Precondition:}}
\algrenewcommand\algorithmicindent{1mm}

\newcommand{\floor}[1]{\left\lfloor #1 \right\rfloor}

\newcommand{\size}[1]{\ensuremath{\left|#1\right|}}

\newcommand{\imply}{\supset}

\newcommand{\mil}{\ensuremath{\mathbf{M}_{\imply}}}

\newcommand{\SP}{$\mathcal{E}\mathcal{B}$}

\newcommand{\card}[1]{card(#1)}

\newboolean{numberspec}
\setboolean{numberspec}{false}
\newcounter{specline}

\newcommand{\spec}[3][{}]{%
   \setcounter{specline}{0}%
   \ifthenelse{\equal{#1}{numbers}}%
      {\setboolean{numberspec}{true}}%
      {\setboolean{numberspec}{false}}%
   \ensuremath{
      \begin{array}{l}
         \ifthenelse{\equal{#2}{}}
            {\numberedline}
            {#2\nl}
         #3
      \end{array}
   }
}

{\begin{myalgo}[#1]
\centering
\begin{minipage}{#2}
\begin{algorithm}[H]}%
{\end{algorithm}
\end{minipage}
\end{myalgo}}

\input{tcilatex}

\title{\bf Exponentially Huge Natural Deduction proofs are Redundant: Preliminary results on $M_{\imply}$}          

\author{
{\bf  Edward Hermann Haeusler$^\dagger$}\\
$^\dagger$Departamento de Inform\'{a}tica - PUC-Rio\\
}

\date{}

\begin{document}

\maketitle                        %%%% To set Title and Author names.

\begin{abstract}

  We estimate the size of a labelled tree by comparing the amount of (labelled) nodes with the size of the set of labels. Roughly speaking, a exponentially big labelled tree, is any labelled tree that has an exponential gap between its size, number of nodes, and the size of its labelling set. The number of sub-formulas of any formula is linear on the size of it, and hence any exponentially big proof has a size $a^n$, where $a>1$ and $n$ is the size of its conclusion. In this article, we show that the linearly height labelled trees whose sizes have a exponential gap with the size of their labelling sets posses at least one sub-tree that occurs exponentially many times in them.
  Natural Deduction proofs and derivations in minimal implicational logic ($\mil$) are essentially labelled trees. By the sub-formula principle any normal derivation of a formula $\alpha$ from a set of formulas $\Gamma=\{\gamma_1,\ldots,\gamma_n\}$ in $\mil$, establishing $\Gamma\vdash_{\mil}\alpha$, has only sub-formulas of the formulas $\alpha,\gamma_1,\ldots,\gamma_n$ occurring in it. By this relationship between labelled trees and derivations in $\mil$, we show that any normal proof of a tautology in $\mil$ that is exponential on the size of its conclusion has a sub-proof that occurs exponentially  many times in it. Thus, any normal and linearly height bounded proof in $\mil$ is inherently redundant. Finally, we  discuss how this redundancy provide us with a highly efficient compression method for propositional proofs. We also provide some examples that serve to convince us that exponentially big proofs are more frequent than one can imagine.   
  \end{abstract}
  
\section{Introduction}

Natural Deduction system, as conceived by Gentzen (\cite{Gentzen1936}),  is given by a set of rules that settle the concept of a deduction for some (logic) language. 
The system of Natural Deduction, as used and considered here, is determined by the logic language and this set of rules also called inference rules. Language and inference rules can be viewed as a {\em logical calculus}, as defined by Church (\cite{Church}). In contrast with the main formulations of logical calculus for some logics by Hilbert (\cite{HilbertHis}), Natural Deduction does not have axioms. Moreover, Natural Deduction implements in the level of the logical calculus the  {\em (meta)theorem of  deduction}, namely from $\Gamma, A\vdash A\imply B$, employing the discharging mechanism. The introduction rule for the $\imply$ application, as it follows, shows how this discharging mechanism implements in the logic calculus the {\em deduction theorem}.

\begin{prooftree}
  \AxiomC{[$A$]}
  \noLine
  \UnaryInfC{$\Pi$}
  \noLine
  \UnaryInfC{$B$}
  \RightLabel{$\imply$-Intro}
  \UnaryInfC{$A\imply B$}
\end{prooftree}

The embracing of some of the formulas $A$ in the derivation $\Pi$ of $B$ from $A$ by []s plays the role of discharging these formula occurrences from $\Pi$. Embracing a formula occurrence means that from the application of the $\imply$-Intro rule that embraces this occurrence of the $\imply$-Intro down to the conclusion of the derivation the inferred formulas do not depend anymore on these embraced occurrences of $A$.
The choice of the formulas that an application of a $\imply$-Intro embraces, as a consequence of an application of a $\imply$-intro rule, is arbitrary. The range of this choice goes from every occurrence of $A$ until none of them. The following derivations show two different ways of deriving $A\imply (A\imply A)$. Observe that in both deductions or derivations, we use numbers to indicate which is the application of the $\imply$-Intro that discharged the marked formula occurrence. For example, in the right derivation, the upper application discharged the marked occurrences of $A$, while in the left derivation, it is the lowest application that discharges the formula occurrences $A$. There is a third derivation that both applications do note discharge any $A$, and the conclusion $A\imply(A\imply A)$ keep depending on $A$. This third alternative appears in figure~\ref{third}. Natural Deduction systems can provide logical calculi without any need to use axioms. In this article, we focus on the system formed only by the $\imply$-Intro rule and the $\imply$-Elim rule, as shown below, also known by {\em modus ponens}. The logic behind this logical calculus is the purely minimal implicational logic, $M_{\imply}$.

\begin{prooftree}
  \AxiomC{$A$}
  \AxiomC{$A\imply B$}
  \RightLabel{$\imply$-Elim}
  \BinaryInfC{$B$}
\end{prooftree}

One thing to observe is that we can substitute liberal discharging mechanism by a greedy discipline of discharging that discharges every possible formula occurrence whenever the $\imply$-Intro is applied. Observe that, in this case, the derivation in figure~\ref{third} would not be possible anymore. Completeness regarding derivability would be lost. However, when considering proofs, i.e.,  derivations with no assumption undischarged, the greedy version of the $\imply$-Intro is enough to ensure the demonstrability of valid formulas.

\begin{prooftree}
  \AxiomC{$[A]^1$}
  \UnaryInfC{$A\imply A$}
  \RightLabel{$\;^1$}
  \UnaryInfC{$A\imply (A\imply A)$}
  \AxiomC{$[A]^1$}
  \RightLabel{$\;^1$}
  \UnaryInfC{$A\imply A$}
  \UnaryInfC{$A\imply (A\imply A)$}
  \noLine
  \BinaryInfC{}
\end{prooftree}

\begin{figure}[H]
\begin{prooftree}
  \AxiomC{$A$}
  \UnaryInfC{$A\imply A$}
  \UnaryInfC{$A\imply (A\imply A)$}
\end{prooftree}
\caption{Two vacuous $\imply$-Intro applications}\label{third}
\end{figure}

With the sake of providing simpler proofs of our results, we take Natural Deduction as trees. From any derivation in ND, there is a binary tree having nodes labelled by the formulas and edges linking premises to conclusion, such that the root of the tree would be the conclusion of the derivation, and the leaves are its assumptions. For example, the derivation in figure~\ref{ex:derivacao} has the tree in figure~\ref{ex:tree} representing it.  The set of formulas that the labelling formula of $u$ depends on the labelling formula of $v$ labels the edge from a node $v$ to a node $u$. This set of formulas is called the dependency set of the label of $u$ from the label of $v$. In this way, the $\imply$-intro, in fact, its greedy version,  removes the discharged formula from the dependency set, as shown in figure~\ref{ex:tree}. Note that because of this labelling of edges by dependency sets, we need one more extra edge and root node to be labelled as the dependency set of the conclusion. That is the reason for the edge linking the conclusion to the dot in figure~\ref{ex:tree}.

\begin{figure}[H]
  \begin{prooftree}
  \AxiomC{$[A]^{1}$}
  \AxiomC{$A\imply B$}
  \BinaryInfC{$B$}
  \AxiomC{$B\imply C$}
  \BinaryInfC{$C$}
  \UnaryInfC{$\LeftLabel{1}A\imply C$}
  \end{prooftree}
\caption{A derivation in $M_{\imply}$}\label{ex:derivacao}
\end{figure}

\begin{figure}[H]
      \begin{tikzpicture}[shorten >=1pt,node distance=1cm,auto]%,on grid
\tikzstyle{state}=[shape=circle,thick,minimum size=1.5cm]
\node[state] (ground) {.};
\node[state, above of=ground] (AC) {$A\imply C$};
\node[state, above of=AC] (C) {$C$};
\node[state, above left of=C] (B) {$B$};
\node[state, above left of=B] (A) {$A$};
\node[state, above right of=B] (AB) {$A\imply B$};
\node[state, above right of=C] (BC) {$B\imply C$};
\path [draw,->] (ground) edge node[midway] {{\tiny $\{A\imply B, B\imply C\}$}} (AC);
\path [draw,->] (AC) edge node[midway] {{\tiny $\{A,A\imply B, B\imply C\}$}} (C);
\path [draw,->] (C) edge node[midway,right] {{\tiny $\{B\imply C\}$}} (BC);
\path [draw,->] (C) edge node[midway, left] {{\tiny $\{A,A\imply B\}$}} (B);
\path [draw,->] (B) edge node[midway, right] {{\tiny $\{A\imply B\}$}} (AB);
\path [draw,->] (B) edge node[midway, left] {{\tiny $\{A\}$}} (A);
    \end{tikzpicture}
\caption{The tree representing derivation in figure~\ref{ex:derivacao}}\label{ex:tree}
\end{figure}
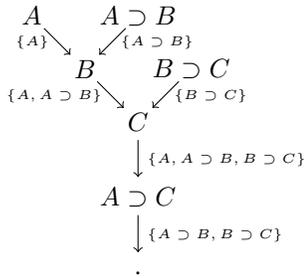

As a matter of computational representation of ND proofs as trees,  we use bitstrings induced by an arbitrary linear ordering of formulas in order to have a more compact representation of the dependency sets. Taking into account that only subformulas of the conclusion can be in any dependency set, we only need bitstrings of the size of the conclusion of the proof. In figure~\ref{ex:bitstring} we show this final form of tree representing the derivation in figure~\ref{ex:derivacao} and~\ref{ex:tree}, when the linear order $\prec$ is $A\prec B\prec C\prec A\imply B\prec B\imply C\prec A\imply C$. 

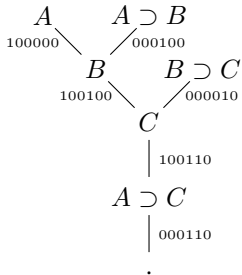
\begin{figure}[H]
    \begin{tikzpicture}[shorten >=1pt,node distance=1cm,auto]%,on grid
\tikzstyle{state}=[shape=circle,thick,minimum size=1.5cm]
\node[state] (ground) {.};
\node[state, above of=ground] (AC) {$A\imply C$};
\node[state, above of=AC] (C) {$C$};
\node[state, above left of=C] (B) {$B$};
\node[state, above left of=B] (A) {$A$};
\node[state, above right of=B] (AB) {$A\imply B$};
\node[state, above right of=C] (BC) {$B\imply C$};
\path [draw] (ground) edge node[midway] {{\tiny $000110$}} (AC);
\path [draw] (AC) edge node[midway] {{\tiny $100110$}} (C);
\path [draw] (C) edge node[midway,right] {{\tiny $000010$}} (BC);
\path [draw] (C) edge node[midway, left] {{\tiny $100100$}} (B);
\path [draw] (B) edge node[midway, right] {{\tiny $000100$}} (AB);
\path [draw] (B) edge node[midway, left] {{\tiny $100000$}} (A);
    \end{tikzpicture}
\caption{Tree with bitstrings representing derivation in figure~\ref{ex:derivacao}}\label{ex:bitstring}
\end{figure}

%\tdmargin{Colocar aqui um exemplo de arvore derivada de prova com labels vindo de $\lambda$:{\color{white} FAZER e lembrar que vai ser apontado na def de $F_{\mathcal{S},\Phi}$ Precisa mesmo ???}}   

%\td{Aqui deve vir a mensagem do artigo. Comentar o fato do principio da subformula fazer com que provas grandes sejam altamente redundantes.}

Proof-theory is the branch of logic, the foundation of Mathematics and Computer Science that studies proofs. It has a well-established set of results\footnote{The proof of the consistency of Arithmetic by Gentzen in the '30s is one of the champion results} and tools devoted to formal proofs and consistency proofs of formalized mathematical theories and metatheories.
However, because of the scope of this paper, we only briefly list the minimal results and definitions in the next section, such that the reader can make the connections between $EOL$-trees and Natural Deduction proofs need to understand the main result of this article. We explain briefly in the next paragraph, the intuition that motivates our result.

Roughly, the sub-formula principle for a logic $\mathcal{L}$ states that all we need to prove a tautology is inside itself. That is, without loss of generality (w.l.g.), If $\alpha$ is a tautology, then there is a proof of $\alpha$ using only sub-formulas of $\alpha$ in it. It is a corollary of the Normalization theorem for Natural Deduction, a central result, and a tool of proof-theory. Well, \mil satisfies the normalization and hence the sub-formula principle. We note that the amount of sub-formulas of any formula is linear on its size. Instead of the denomination `exponentially big proofs' we use the denomination `huge proofs' simply. Thus, a  proof is considered huge whenever its size is larger than or equal to any exponential on the size of its conclusion\footnote{If we follow Cook-Karp conjecture that says that computationally easy to verify and compute objects are of polynomial size, huge proofs include the hard proofs for verification, namely, the super-polynomial ones}. Thus, if a proof is exponentially big, i.e., {\bf huge}, its corresponding labeled tree is also exponential. That is, it is at least of size $a^n$, $a>1$ and $n$ is the size/length of the formula that labels its root. We remind that each sub-formula is a possible label node in the tree. We have then that a exponentially big proof is labeled with linearly many, regarding the size/length of the string that labels its root. This configuration allows us to say that at least one label repeats exponentially many times in the tree under the additional  consideration that the tree is linearly-height bounded. We show that this repetition happens in a way that a sub-tree is repeated exponentially many times.

The additional hypothesis on the linear bound on height of the proof of \mil tautologies can be taken into account without loss of generality if we consider the complexity class $CoNP$ (see appendix~\ref{appendix:hamiltonian}). Moreover, in \cite{Studia}, we show that any tautology in \mil has a Natural Deduction normal proof of height bound by the square of the size of this tautology. However, it is not easy to extend the  reasoning in this article to the case of polynomially height-bound trees..  

In the next section, section~\ref{terminology},  we provide the main terminology and definitions used in the article. Section~\ref{sec:HugeProofs} examines three examples of classes of huge proofs and shows concrete cases of the redundancy in huge proofs. Section~\ref{sec:Main} states and proves the main lemma, and finally, section~\ref{sec:Conclusion} discusses the consequences of what we show here for compressing propositional proofs.

\section{Terminology and definitions}\label{terminology}

Following the usual terminology in Natural Deduction and proof-theory, we briefly describe what we use in this article.

The left premise of a $\imply$-Elim rule is called a minor premise, and the right premise is called the major premise. We should note that the conclusion of this rule, as well as its minor premise, are sub-formulas of its major premise. We also observe that the premise of the $\imply$-Intro is the sub-formula of its conclusion. A derivation is a tree-like structure built using $\imply$-Intro and $\imply$-Elim rules. We have some examples depicted in the last section. The conclusion of the derivation is the root of this tree-like structure, and the leaves are what we call top-formulas. A proof is a derivation that has every top-formula discharged by a $\imply$-Intro application in it. The top-formulas are also called assumptions. An assumption that it is not discharged by any rule $\imply$-Intro in a derivation is called an open assumption. If $\Pi$ is a derivation with conclusion $\alpha$ and $\delta_1,\ldots,\delta_n$ as all of its open assumptions then we say that $\Pi$ is a derivation of $\alpha$ from $\delta_1,\ldots,\delta_n$.

\begin{definition}\label{def:Branch} A branch of a derivation or proof $\Pi$ is any sequence $\beta_1,\ldots,\beta_k$ of formula occurrences in $\Pi$, such that:
\begin{itemize}
\item $\delta_1$ is a top-formula, and;
\item For every $i=1,k-1$,  either $\beta_i$ is a $\imply$-Elim major premise of $\beta_{i+1}$ or $\beta_i$ is a $\imply$-Intro premise of $\beta_{i+1}$, and;
\item $\delta_k$ either is the conclusion of the derivation or the minor premise of a $\imply$-Elim.
\end{itemize}
\end{definition}

A normal derivation/proof in \mil is any derivation that does not have any formula occurrence that is simultaneously a major premise of a $\imply$-Elim and the conclusion of a $\imply$-Intro. A formula occurrence that is at the same time  a conclusion of a $\imply$-Intro and a major premise of $\imply$-Elim is called a maximal formula. 
In \cite{Prawitz} there is  the proof of the following theorem for the Natural Deduction for the full\footnote{The full propositional fragment is $\{\lor, \land, \imply, \neg, \bot\}$} propositional fragment of minimal logic. The proof of this theorem uses a strategy of application of a set of reduction rules that eliminates maximal rules. It is out of the scope of this article to provide more details on the proof of the normalization theorem.     

\begin{theorem}[Normalization]\label{theo:Normalization}
  Let $\Pi$ be a derivation of $\alpha$ from $\Delta=\{\delta_1,\ldots, \delta_{n}\}$. There is a normal proof $\Pi^{\prime}$ of $\alpha$ from $\Delta^{\prime}\subseteq\Delta$.
\end{theorem}

In any normal derivation/proof, the format of a branch is essential and provides worth information on why huge proofs are redundant, as we will see in the next sections. Since no formula occurrence can be a major premise of $\imply$-Elim and conclusion of a $\imply$-Intro rule in a branch we have that the conclusion of a $\imply$-Intro can only be the minor premise of a $\imply$-Elim or it is not a premise at all. In this last case, it is the conclusion of the derivation. In any case, it is the last formula in the branch. Thus, any conclusion of a $\imply$-Intro has to be a premise of a $\imply$-Intro. Hence, any branch in a normal derivation is divided into two parts (possibly empty). The E-part that starts the branch with the top-formula and every formula occurrence in it is the major premise of a $\imply$-Elim. There is a formula occurrence that is the conclusion of a $\imply$-Elim and premise of a $\imply$-Intro rule that is called minimal formula of the branch. The minimal formula starts the I-part of the branch, where every formula is the premise of a $\imply$-Intro, excepted the last formula of the branch. From the format of the branches, we can conclude that the sub-formula principle holds for normal proofs in Natural Deduction for \mil, in fact, for many extensions of it.

\begin{corollary}[Sub-formula principle]
  Let $\Pi$ be a normal derivation of $\alpha$ from $\Delta=\{\delta_1,\ldots,\delta_m\}$. It is the case that for every formula occurrence $\beta$ in $\Pi$, $\beta$ is a sub-formula of either $\alpha$ or of some of $\delta_i$.
\end{corollary}

This corollary ensures that without loss of generality, any Natural Deduction proof of a \mil tautology has only sub-formulas of it occurring in it. Normal proofs/derivations offer the $EOL$-tree abstraction in forms of the trees associated with derivations in Natural Deduction for \mil as we show in the sequence. The definition of $EOL$-tree facilitates the proof of the main result of this article. With labelled trees, we can focus on the combinatorial aspects rather than the proof-theoretical.

We assume the standard definition of a tree and a (possibly) incomplete binary tree. The size of a, possibly incomplete,  tree $\langle V, E\rangle$ is $\size{V}$, the number of vertexes of the tree. As a tree is a simple graph, the number of edges is upper-bounded by $\size{V}^{2}$.  The root of a tree is the unique node $r\in V$, such that, there is no $v\in V$, such that $\langle r,v\rangle\in E$. The level of a node $v$, $lev(v)$, in a tree $\mathcal{T}=\langle V, E\rangle$, is the number of nodes in a path from $v$ to the root of the tree. This can be defined in a recursive/inductive way as: (basis) $lev(r)=0$; (rec) if $lev(u)=n$ and $\langle u,v\rangle\in E$ then $lev(v)=n+1$.

The following definition is motivated by Natural Deduction, in \mil, derivation trees, as already said.  
When applying lemma~\ref{lemma:Main} below, we can think of Natural Deduction derivation trees in $M_{\imply}$. We reinforce that a vertex with two children plays the role of an instance of a $\imply$-Elimination role application having it as the conclusion, and a vertex with one child plays the role of an instance of $\imply$-Introduction with it as the conclusion. The leaves are either hypothesis, also called open assumptions, of the derivation or discharged assumptions. In this section scope, there is no representation for the discharging function attached to each instance of $\imply$-Introduction application.   In the concrete case that the labels of the nodes, set $B$ below, are formulas and the order is the sub-formula ordering between propositional implicational formulas\footnote{We advise the reader to not confuse this partial ordering abstracted from the sub-formula ordering, with the linear and arbitrary ordering that is used to generate the bitstrings mentioned at the introduction}. However, for the sake of simplicity, we consider any partial order in the proof of the main results.
%\td{Cuidado aqui para verificar se a versao abstrata vale mesmo}

  %\todo[linecolor=red,backgroundcolor=red!25,bordercolor=red]{It remains to define what is an ordered tree, with the three types of edges $E_{uni}$, $E_{left}$ and $E_{right}$.}
  
\begin{definition}[{\bf EOL-Binary tree}]\label{EOLTree} An edge-ordered-labeled labeled binary tree $\mathcal{T}$ is a structure $\langle V,E_{L}\cup E_{R}\cup E_{U},l,B\rangle$, where:
\begin{enumerate}
\item $\langle V,E_{L}\cup E_{R}\cup E_{U}\rangle$ is a (possibly incomplete) binary tree and;
\item $l:V\rightarrow B$ is the labeling function, with $B$ a finite and partially ordered set of labels, with a partial operation $\odot$ and;
\item $E_{L},  E_{R}$ and $E_{U}$ are mutually disjoint, and;
\item Whenever $\langle v,v_1\rangle\in E_R$ and $\langle v,v_2\rangle\in E_R$ then $v_1=v_2$;
\item Whenever $\langle v,v_1\rangle\in E_L$ and $\langle v,v_2\rangle\in E_L$ then $v_1=v_2$;
\item $\langle v,v_1\rangle\in E_{L}$, if and only if, $\langle v,v_2\rangle\in E_{R}$, $v_2\neq v_1$, and;
\item If $\langle v,v^{\prime}\rangle\in E_{U}$ and $\langle v,v^{\prime\prime}\rangle\in E_{U}$ then $v^{\prime}=v^{\prime\prime}$, and;
\item\label{item:one-to-one} If $\langle v,v_1\rangle\in E_{L}$ and $\langle v,v_2\rangle\in E_{R}$ then  $l(v)\prec_{B} l(v_2)$ and $l(v_1)\prec_{B} l(v_2)$
  and for each $l(v_2)$, such that, $l(v)\prec_{B} l(v_2)$, there is only one $b\in B$, such that $b\odot l(v)$ is equal to $l(v_2)$. This $b$ should be equal to $l(v_1)$;
\item\label{item:unique-U} If $\langle v,v^{\prime}\rangle\in E_{U}$ then there is $q\in B$, such that, $q\odot l(v^{\prime})=l(v)$ and, of course, $l(v^{prime})\prec_{B} l(v)$
\end{enumerate}
\end{definition}

In this article, we are interested in the kind of EOL-tree that is linearly-bounded on the height regarding the size of its labelling set. That is, $T$ is linearly bounded on the height when $h(T)\leq k\times\card{B(T)}$, for some $0<k\in\mathbb{R}$ where $h(T)$ is the height of the tree $T$\footnote{We can say also that $h(T)\in\mathcal{O}(B(T))$}. We call these trees linearly height $B$-labelled trees or linear-height $EOL$-trees. In section~\ref{sec:Main}, we prove that for any linearly-height $B$-labelled tree of exponential or bigger size, there is a tree that occurs exponentially many times as a subtree in it. The content of  lemma~\ref{lemma:Main} states this. In the sequel, we provide more definitions that we use. In section~\ref{sec:Main}, we also comment on the non-triviality of the extension of this main result to polynomially height bounded $EOL$-trees. 

A skeletal-tree is defined as any non-empty $B$-labeled edge-labelled tree with edge labels {\em U}, {\em L} and {\em R}. In the sequel, we ask the reader to remind herself (himself) the concept of injective tree-mapping from labelled trees into labelled trees. (see any book on graph theory or theory of computation, for example, \cite{arora}). 
 
\begin{definition}[{\bf Skeletal-tree occurrence}]
  A  skeletal-tree instance $\mathcal{S}$ occurs  in a $B$-labeled tree $\mathcal{T}=\langle V, E, l, B\rangle$, iff,  there is an injective $B$-labeled tree mapping $f$ from $\mathcal{S}$ into $\mathcal{T}$, such that, for each $v, u\in V_{\mathcal{S}}$, if $\langle v,u\rangle\in E_{\lambda}$ then $\langle f(v),f(u)\rangle\in E_{\lambda}$, $\lambda=U,L,R$; $l(v)=l(f(v))$ and $l(u)=l(f(u))$.
%   \todo[linecolor=red,backgroundcolor=red!25,bordercolor=red]{Need to be defined}
\end{definition}

Whenever a skeletal tree instance  $\mathcal{S}$ occurs in a tree $\mathcal{T}$, we say that there is a sub-tree of the skeletal form $\mathcal{S}$ in the tree. We also say simply that $\mathcal{S}$ is a sub-tree occurring in $\mathcal{T}$. We can conclude that any sub-tree of an instance of a skeletal occurring in a tree is also a skeletal sub-tree instance of this first tree. A skeletal-tree occurrence/instance is full whenever if it cannot be extended to other sub-tree by adding any contiguous vertex from the tree to it. Sometimes we use the term ''Skeletal-tree instance', instead of ``skeletal-tree occurrence''. When a Skeletal-tree instance $\mathcal{Y}$ occurs in a tree $\mathcal{T}$ and $\mathcal{Y}$'s root is at level $k$ then we say that $\mathcal{Y}$ occurs at level $k$ in $\mathcal{T}$.

 Concerning the computational complexity of propositional proofs, we count the size of a proof as to the number of symbol occurrences used to write it. If we put all the symbol occurrences used to write a Natural Deduction derivation $\Pi$ side by side in a long string then the size of the derivation, denoted by $\size{\Pi}$, is the length of this string. The function $\size{\;}:Strings\longrightarrow \mathbb{N}$, the size-of-string function, denotes the mapping of strings to their corresponding sizes\footnote{Some authors use the term {\em lenght} instead of {\em size}}. For derivations $\alpha$ from $\Delta=\{\delta_1,\ldots, \delta_n\}$ we estimate the complexity of the derivation by means of a function of $\size{\alpha}+\sum_{i=1,n}\size{\delta_i}$ into the size of the derivation itself. Thus, we should not take the complexity of a derivation individually. It is taken together with the set of all derivations.

%\todo[linecolor=red,backgroundcolor=red!25,bordercolor=red]{Reescrever este paragrafo sobre CC {\color{white} Paragrafo Importante para fazer somente depois do resto do artigo pronto}} 

A set $\mathcal{S}$ of $EOL$-trees is unlimited, if and only if, for every $n>0$ there is $T\in \mathcal{S}$, such that, $\size{T}>n$. 

\begin{definition}\label{SP:local}
  An unlimited set $\mathcal{S}$ of EOL-trees is huge (exponentially big or $\mathcal{E}\mathcal{B}$ for short) iff there are $a\in\mathbb{R}$, $a>1$, $n_0,p\in\mathbb{N}$, $p>1$, $c\in\mathbb{R}$, $c>0$, such that, for every $n>n_0$ and for every $T\in\mathcal{S}$, if $\card{B(T)}= n$ then $\size{T}\ge c\times a^{n^p}$.
\end{definition}

If $A$ is a set then we use $\card{A}$ to denote the number of elements in $A$.

In this article, we use an alternative, equivalent, and more suitable definition for use in the demonstration of our result than the above one. We use the following auxiliary definitions to define it. 

\begin{definition}
  Let $\mathcal{S}$ be an unlimited set of $EOL$-trees and $\size{\;}$ the size-of-string function. The function $len_{\mathcal{S}}:\mathcal{S}\rightarrow \mathbb{N}$ is the defined as $len_{\mathcal{S}}(T)=\size{T}$.
\end{definition}

In the definition above, we advise the reader that the size of the alphabet used to write the strings is at least 2. Unary strings cannot be consistently used in computational complexity estimations, since its use trivializes\footnote{If there is a NP-complete Formal Language $L\subseteq\Sigma^{\star}$, where $\Sigma$ is a singleton, then $NP=P$, see for example \cite{Creszenci} (theorem 5.7, page 87)} the conjecture $NP=P$. We use to call an alphabet reasonable whenever it has at least two symbols. 

\begin{definition}
  A function  $f:\mathbb{N}\longrightarrow\mathbb{N}$ is exponential or bigger if and only if there are $a\in\mathbb{R}$, $a>1$, $n_0,p\in\mathbb{N}$, $c\in\mathbb{Q}$, $p>1$, $c>0$, such that, $\forall n>n_0$, $f(n)\ge c\times a^{n^p}$.
\end{definition}

Technically, the above definition says that a funcion is exponential or bigger whenever it has a tight exponential lower bound. 

Consider a property $\Phi(x)$ on $EOL$-trees. This property is used to select, from a set $\mathcal{S}$, all the $EOL$-trees satisfying it. This defines a subset $\{T\in \mathcal{S}:\Phi(T)\}$ of $\mathcal{S}$. As an example we can set a particular $\Phi_{\Gamma,\alpha}(x)$, where $\Gamma$ is a set of labels and $\alpha$ is a label, to be true only on $EOL$-trees $T$, such that $leaves(T)=\Gamma$ and $r(T)=\alpha$. Thus, given a set $\mathcal{S}$ of $EOL$-trees, the set $\{T\in \mathcal{S}:\Phi_{\Gamma,\alpha}(T)\}$ is the subset of all trees from $\mathcal{S}$ that have $\Gamma$ as labelling the leaves and $\alpha$ labelling the respective root of each of them. We use properties as $\Phi_{\Gamma,\alpha}(x)$  to specify the set of all trees that correspond to Natural Deduction derivations of a formula $\alpha$ from a set of hypothesis $\Gamma$. We further refine this to get the set of all minimal trees (derivations) of $\alpha$ from $\Gamma$. For example
\[
Min_{\mathcal{S}}(\Gamma,\alpha)=\{T\in\mathcal{S}:\mbox{$\Phi_{\Gamma,\alpha}(T)\;\land\;\forall T^{\prime}(\Phi_{\Gamma,\alpha}(T^{\prime})\;\rightarrow\; \size{T}\leq\size{T^{\prime}})$}\}
\]
is the set of all smallest  $EOL$-trees satisfying $\Phi_{\Gamma,\alpha}(x)$. We can see these $EOL$-trees as Natural Deduction derivations in $\mil$, having then the set of all smallest derivations of $\alpha$ from $\Gamma$ in $\mil$. In the general case, where the predicate $\Phi(x)$ is arbitrary, we denote the set above by $Min_{\mathcal{S}}(\Phi)$, that is:
\[
Min_{\mathcal{S}}(\Phi)=\{T\in\mathcal{S}:\mbox{$\Phi(T)\;\land\;\forall T^{\prime}(\Phi(T^{\prime})\;\rightarrow\; \size{T}\leq\size{T^{\prime}})$}\}
\]

\begin{definition}\label{def:FuncaoF}
  Let $\mathcal{S}$ be an unlimited set of $EOL$-trees. Let $\Phi(x)$ represent a property on $EOL$-trees of $\mathcal{S}$ and let $\Phi_{\mathcal{S},m}(x)$ be defined as $(x\in\mathcal{S}\;\land\;\Phi(x)\;\land\;\card{B(x)}\leq m)$ with $0<m\in\mathbb{N}$. We define the function $F_{\mathcal{S},\Phi}:\mathbb{N}\longrightarrow\mathbb{N}$ that associates do each natural number $m$ the least tree satisfying $\Phi_{\mathcal{S},m}(x)$.
  \[
  F_{\mathcal{S},\Phi}(m)=\left\{\begin{array}{ll} 0 & \mbox{if $m=0$} \\
                                      \size{Min_{\mathcal{S}}(\Phi_{\mathcal{S},m})} & \mbox{if $m>0$}               
  \end{array}\right.
  \]
\end{definition}

We point out that depending on $\Phi$, the above function $F_{\mathcal{S},\Phi}$ can be quite uninteresting. For example, if $\Phi$ is satisfiable by every tree in $\mathcal{S}$ then  $F_{\mathcal{S},\Phi}(m)=1$, for every $m>0$. Any tree $T$ with only one node,  such that, it is labelled by any element of $B(T)$, is a smallest\footnote{We do not take the null tree in this work since it represents no meaningful representation of data in our case} tree that satisfies $\Phi$. On the other hand, we can have $\Phi_{\mathcal{S},\Phi}(m)$ true only when $T$ is a tree that represents a proof of a \mil tautology $\alpha$, $m=\size{\alpha}$ and $\Phi(T)$ is true whenever $T$ is a tree representing a proof of $\alpha$.

The following proposition points out an alternative and more adequate definition for a family of exponential or bigger sized trees as already previously mentioned. Observe that if $\mathcal{A}$ is the set of all $EOL$-trees and $\Phi(x)$ is a property defining a subset $\mathcal{S}$ of $\mathcal{A}$ and $\Phi_{\mathcal{S},m}$ is defined as in definition~\ref{def:FuncaoF} then $\mathcal{S}=\Phi(\mathcal{A})=\bigcup_{m\in\mathbb{N}}\Phi_{\mathcal{A},m}(\mathcal{A})$. The reader should note that we use $\Phi(\mathcal{A})$ as an abbreviation of $\{T:\mbox{$\Phi(T)\land T\in\mathcal{A}$}\}$. Observing what is discussed in the last paragraphs, we have the following proposition.

\begin{proposition}\label{prop:FuncaoF}
  Let $\mathcal{S}\subset\mathcal{A}$ be an unlimited set of $EOL$-trees. Let $\Phi(x)$ be the defining property of $\mathcal{S}$. We have then that  $\mathcal{S}$ is \SP\; if and only if $F_{\mathcal{A},\Phi}$ is a exponential or bigger function from $\mathbb{N}$ in $\mathbb{N}$.
  \end{proposition}

We note that the size of a $EOL$-tree considers the representation of it when coded in a string under a reasonable alphabet\footnote{An alphabet is reasonable if it has more than one symbol}.

The following definition is quite useful in the statement of the main lemma of this article.

\begin{definition}[Set of nodes at level $i$ labeled with $q$]
  Given an $EOL$-tree $T=\langle V,E_{L}\cup E_{R}\cup E_{U},l,B\rangle$, a natural number $i\in\mathbb{N}$ and a label $q\in B(T)$. We use the notation $V_{T}^{i,q}$ to denote:
\[
\{v\in V(T):\mbox{$lev_{T}(v)=i$ and $l_{T}(v)=q$}\}
\]
  \end{definition}

\begin{definition}\label{def:occ}
  Let $\mathcal{S}$ be an unlimited set of $EOL$-trees and $\size{\;}$ the size-of-string function. The functions $occ_{\mathcal{S}}^{i,q}:\mathcal{S}\rightarrow \mathbb{N}$, $i\in\mathbb{N}$, and $q\in SYMB$, are defined below, where $\bigcup_{T\in\mathcal{S}}B(T)$ is the symbol set\footnote{We allow the existence of an infinite set of symbols, but this is not essencial at this point of our formalization} that labels the trees in  $\mathcal{S}$:
  \[
  occ^{i,q}_{\mathcal{S}}(T)=\size{V_{T}^{i,q}}
  \]
\end{definition}

Note that for any $T\in\mathcal{S}$, if $i>h(T)$ then $occ^{i,q}_{\mathcal{S}}(T)=0$ and that for any $q\not\in B(T)$, $occ^{i,q}_{\mathcal{S}}(T)=0$ too. 

%% Let $\mathcal{A}$ be the set of all $EOL$-trees associated with the Natural Deduction minimal implicational logic derivations. Consider that $\Phi(T)$ holds only for  $T$, such that, $T$ is an $EOL$-tree derived from a N.D. proof of $r(T)$, i.e., the conclusion of $T$. We have then the following proposition that we proved by using the fact that \mil~ is $PSPACE$-complete and that $NP$ is the class of all the decision problems with positive polynomially sized and verifiable certificates, as it is discussed in \cite{Ladner} and \cite{Svejdar}.

%% \begin{proposition}
%%   $NP=PSPACE$, if and only if, $F_{\mathcal{S},\Phi}$ is polynomially upper bounded
%%   \end{proposition}

The following lemma is essential in the proof of the results shown in this article.

\begin{lemma}[Super-Exponential Pigeonhole Principle]\label{lemma:SPPHP}
  Let $k:\mathbb{N}\longrightarrow\mathbb{N}$ be a polynomial and  let $f_i:\mathbb{N}\longrightarrow \mathbb{N}$ be a function, for each $i\in\mathbb{N}$.Consider $m\in\mathbb{N}$ and the function $g_m:\mathbb{N}\longrightarrow\mathbb{N}$ defined as:
  \[
  g_m(n)=\sum_{i\leq k(m)} f_i(n)
  \]
  For each $m$, if  $g_m$ is \SP\ on $m$ then there is $0\leq j\leq k(m)$, such that, $f_j$ is \SP,  on $m$, too. 
\end{lemma}

\begin{proof}
Consider some $m\in\mathbb{N}$ and suppose that there is no $j$, $0\leq j\leq k(m)$, such that $f_j$ is exponential or bigger than it on $m$. Thus, for every $j$, $f_j$ is polynomially bounded. Since any polynomial sum of polynomials is a polynomial too, we conclude that $g_m$, for this $m$, is polynomially bounded, that is a contradiction. 
  
  \end{proof}

%\pontoparada{REVISAO NO GRAMMARLY FEITA DO ABSTRACT ATE AQUI}

\section{Some useful properties in normal proofs in \mil}

\begin{definition} Let $\alpha$ be a formula in \mil, the abstract syntax tree of $\alpha$ is the tree $T_{\alpha}$ defined by induction as:
  \begin{itemize}
  \item If $\alpha$ is a propositional letter $A$ then $T_{A}=A$, and;
  \item If $\alpha$ is $\alpha_1\imply\alpha_2$ then $T_{\alpha}$ is
  \begin{tabular}{l}
      \Tree [.$\imply$ $T_{\alpha_{1}}$ $T_{\alpha_{2}}$ ]
  \end{tabular}
  \end{itemize}
\end{definition}

It is easy to see that, for every $\alpha$ in \mil,  $T_{\alpha}$ is a full binary tree\footnote{A binary tree, such that, every node that is not a leaf has two childrens}. For example, below we can find the abstract syntax trees of $((A\imply B)\imply A)\imply A)$ and $(A\imply B)\imply ((B\imply C)\imply (A\imply C))$ respectively.
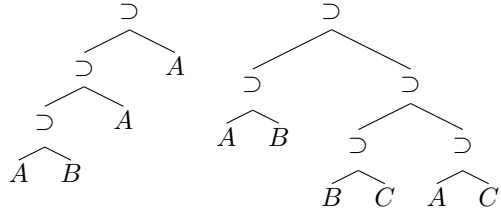
\[
\begin{tabular}{ll}
      \Tree [.$\imply$ [.$\imply$ [.$\imply$ $A$ $B$ ] $A$ ] $A$ ]
      &
      \Tree [.$\imply$ [.$\imply$ $A$ $B$ ] [.$\imply$ [.$\imply$ $B$ $C$ ] [.$\imply$ $A$ $C$ ]]]
\end{tabular}
\]

We remind the reader that $A\imply (B\imply C)$ can be written as $A\imply B\imply C$ under the convention that the parenthesis are grouped to the right. We also observe that the tree representation seems to be fully adequate to be used here in this work on computational complexity analysis of the size of proofs. We only have to observe that for any $\alpha$ the lenght of $\alpha$ is, in general, larger than the number the subformulas of $\alpha$. When $\alpha$ has parenthesis, they hold positions that do not determine any subformula from $\alpha$. On the other hand, each node in $T_{\alpha}$ determines a subformula from $\alpha$. Thus, instead of estimating the size of a proof as a function of the lenght of the a formula $\alpha$, we take in to account the size of the abstract parsing tree $\size{T_{\alpha}}$, i.e., how many nodes it has. 

The height of a tree $T$, denoted by $h(T)$, is the the lenght of the longest path in $T$ from the root to any of its leaves. It is well-known that the height of a balanced tree with size $n$ is $\log(n)$. On the other hand, the longest path from the root to any leaf in a tree $T$ is bounded by $\floor{\frac{\size{T}}{2}}$. For example,  when $T$ is the abstract syntax tree of $a_1\imply (a_2\imply (a_3\ldots (a_n\imply b))\ldots)$, as shown in the tree in figure~\ref{fig:right-unbalenced}. Another thing to observe is that, because of the way the implications are nested by the right, we can consider the right-hand side of the rightmost implication as being a propositional variable. This is summarized in observation~\ref{obs:Height}

\begin{figure}
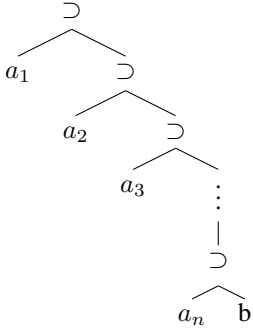

\begin{tabular}{l}
  \Tree [.$\imply$ $a_1$ [.$\imply$ $a_2$ [.$\imply$ $a_3$  [.$\vdots$ [.$\imply$ $a_{n}$ b ]]]]]
\end{tabular}
\caption{A totally right-unbalanced abstract syntax tree}\label{fig:right-unbalenced}
\end{figure}

\begin{figure}
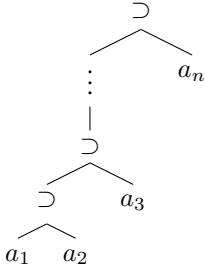

\begin{tabular}{l}
  \Tree [.$\imply$ [.$\vdots$ [.$\imply$ [.$\imply$ $a_1$ $a_2$ ] $a_3$ ] ] $a_n$ ]
\end{tabular}
\caption{A totally left-unbalanced abstract syntax tree}\label{fig:left-unbalenced}
\end{figure}

\begin{observation}\label{obs:Height} For every tree $T$, $\floor{\log(\size{T})}\leq h(T)\leq \floor{\frac{\size{T}}{2}}$
\end{observation}

In any normal proof $\Pi$ of a formula $\alpha$,  the subformulas of it that can be top-formulas of some branch in $\Pi$ are antecedents of $\alpha$. We observe that any formula $\alpha$ in $\mil$ is of the form $\alpha_1\imply(\alpha_2\imply\ldots (\alpha_n\imply q)\ldots)$, with $q$ being a propositional variable. Moreover, the antecents of $\alpha$ are, hence, $\alpha_i$, $i=1,n$. These are the formulas that can be the top-formula of the main branch, due to the fact that in order to prove $\alpha$, some $\imply$-introductions can be needed and one of them has to discharge an assumption $\alpha_j$, for some $j$, that it is top-formula of the main branch. For the secondary branches, each of them is determined by an antecedent of the top-formula $\alpha_j$ of the main branch,  and, that is the minor premise of the correponding $\imply$-Elim application that has a subformula derived from $\alpha_j$, or itself in the topmost $\imply$-Elim application, as major premise.   

In virtue of theorem~\ref{theo:Normalization} and definition~\ref{def:Branch}, any branch in a normal proof has an Elim-part, where there are only $\imply$-Elim rule application rules, and an Intro-part, containing only $\imply$-Intro rule applications. The minimal formula in the branch is the conclusion of the last $\imply$-Elim rule and premisse of the first $\imply$-Intro rule, if any of them. This fact show us that the Elim-part of any branch in a proof of $T_{\alpha}$, i.e. a proof of $\alpha$, is a sequence of subformulas of the respective top-formula of this branch. Any  branch in $T_{\alpha}$ can be identified by its level and its top-formula. Moreover, the sequence of formulas of the E-part of this branch is the sequence of subformulas of the top-formula that appears in the abstract syntax tree $T_{\alpha}$. Thus, we have the following proposition.

\begin{proposition}\label{fact:UniqueTopmost}
  In a normal proof of $\alpha$, the E-part of any branch is uniquely  determined by the formula $\beta$ that is the major premiss of the topmost $\imply$-Elim rule application of this branch. The subtree $T_{\beta}$ of  $T_{\alpha}$ determines each formula occurrence in this E-part of the branch. These formulas are the sequence of right descendents of $\beta$ in $T_{\beta}$.
\end{proposition}

A corollary of the proposition~\ref{fact:UniqueTopmost} is that there is one-to-one correspondence between  sub-formulas of $\alpha$ and possible E-parts of the branches that occurs in any normal proof of $\alpha$. This is stated by the following lemma.

\begin{lemma}\label{lemma:UniqueBranches}
  Let $\Pi$ be a normal proof of $\alpha$. Then each E-part of a branch in $\Pi$ is the sequence of formulas occurring in the right-branch of $T_{\beta}$, where $\beta$ is the major premise of the topmost Elim-rule application in the branch.
\end{lemma}

%% We have to observe that the upper-bound on the lenght of E-parts in normal proofs provide by fact~\ref{fact:Height} can be tighter. From the observation~\ref{obs:Height} we know that the most that $h(T_{\alpha})$ can be is $\floor{\frac{\size{T}}{2}}$. Let us analyse the cases such that $\floor{\log(\size{T_{\alpha}})}<h(T_{\alpha})$ starting with $h(T_{\alpha})=\floor{\frac{\size{T}}{2}}$. These  are of the shape of figure~\ref{fig:right-unbalenced} and of figure~\ref{fig:left-unbalenced}. In both cases, the longest E-parts are of size 1. If $h(T_{\alpha})=\floor{\frac{\size{T}}{2}}$ then in both cases, i.e. figures~\ref{fig:right-unbalenced} and~\ref{fig:left-unbalenced}, $b$ and $a_i$, $i=1,n$, are all propositional variables. We examine both cases in the sequel. 

%% What we see in figure~\ref{fig:right-unbalenced} is a formula that
%% has no antecedent of degree more than 1, so none of them can be used as major premise in an $\imply$-Elim application to prove $\alpha$ by means of a normal proof. They are the only possible formulas to be used in a normal proof of $\alpha$, by the subformula principle and the shape of normal proofs. Thus, if $\alpha$ is a \mil tautology then any normal proof of it has only one branch with no $E$-part, so that $\alpha$ is obtained by a sequence of $\imply$-Intro applications.

%% In figure~\ref{fig:left-unbalenced}, we have the abstract syntax tree of a formula that has only 

One last thing to observe is that the amount of formulas in a normal proof of $\alpha$ that be major premise in a topmost Elim-rule application in a branch is upper-bounded by $\frac{\size{T_{\alpha}}}{3}$, since the major premise of a $\imply$-Elim rule must be of degree at least 3.

\section{Examining some huge normal proofs and their inherent redundancy}\label{sec:HugeProofs}

The following sets can be seen as examples of \SP\; sets of EOL-trees, whenever we  take Natural Deduction derivations in the format of EOL-trees.
%\pontoparada{Digitar a prova indutiva dos niveis de fibonacci e os outros aqui}

\subsection{A family of huge proofs in $\mil$: Fibonacci numbers}

We show here a family of formulas that have only huge proofs as least normal proofs. Consider the following formulas.
\begin{itemize}
\item $\eta = A_1\imply A_2$
\item $\sigma_k = A_{k-2}\imply (A_{k-1}\imply A_{k})$, $k>2$.
\end{itemize}

Any normal proof of $\Pi$ of $A_1\imply A_n$ from $\eta,\sigma_1,\ldots,\sigma_n$ is such that $\size{\Pi}\geq Fibonnaci(n)$.
\begin{prooftree}
\def\defaultHypSeparation{\hskip .05in}
\def\ScoreOverang{1pt}
\AxiomC{$[A_1]$}
\noLine
\UnaryInfC{$A_1\imply A_2$}
% \AxiomC{$A_2$}
% \AxiomC{$A_1$}
\noLine
\UnaryInfC{$A_1\imply(A_2\imply A_3)$}
\noLine
\UnaryInfC{$\Pi_3$}
%\BinaryInfC{$A_2\imply A_3$}
\noLine
\UnaryInfC{$A_3$}
\AxiomC{$[A_1]$}
\AxiomC{$A_1\imply A_2$}
\BinaryInfC{$A_2$}
\AxiomC{$A_2\imply(A_3\imply A_4)$}
\BinaryInfC{$A_3\imply A_4$}
\BinaryInfC{$A_4$}
\AxiomC{$[A_1]$}
\noLine
\UnaryInfC{$A_1\imply A_2$}
% \AxiomC{$A_2$}
% \AxiomC{$A_1$}
\noLine
\UnaryInfC{$A_1\imply(A_2\imply A_3)$}
\noLine
\UnaryInfC{$\Pi_3$}
%\BinaryInfC{$A_2\imply A_3$}
\noLine
\UnaryInfC{$A_3$}
\AxiomC{$A_3\imply(A_4\imply A_5)$}
\BinaryInfC{$A_4\imply A_5$}
\BinaryInfC{$A_5$}
\UnaryInfC{$A_1\imply A_5$}
\end{prooftree}

In general, for each $5\leq k$

\begin{prooftree}
\def\defaultHypSeparation{\hskip .05in}
\def\ScoreOverang{1pt}
\AxiomC{$[A_1]$}
\noLine
\UnaryInfC{$\eta$}
% \AxiomC{$A_2$}
% \AxiomC{$A_1$}
\noLine
\UnaryInfC{$\sigma_3,\ldots,\sigma_{k-1}$}
\noLine
\UnaryInfC{$\Pi_{k-1}$}
%\BinaryInfC{$A_2\imply A_3$}
\noLine
\UnaryInfC{$A_{k-1}$}
\AxiomC{$[A_1]$}
\noLine
\UnaryInfC{$\eta$}
% \AxiomC{$A_2$}
% \AxiomC{$A_1$}
\noLine
\UnaryInfC{$\sigma_3,\ldots,\sigma_{k-2}$}
\noLine
\UnaryInfC{$\Pi_{k-2}$}
%\BinaryInfC{$A_2\imply A_3$}
\noLine
\UnaryInfC{$A_{k-2}$}
\AxiomC{$A_{k-2}\imply (A_{k-1}\imply A_k)$}
\BinaryInfC{$A_{k-1}\imply A_k$}
\BinaryInfC{$A_k$}
\UnaryInfC{$A_1\imply A_k$}
\end{prooftree}

So we have that:
\begin{eqnarray*}
l(\Pi_2) & = & 1 \\
l(\Pi_3) & = & l(\Pi_2) + 1 \\
l(\Pi_k) & = & l(\Pi_{k-2})+l(\Pi_{k-1}) + 2 \\
\end{eqnarray*}
Thus, by a well-known fact about Fibonacci numbers, we see that:
\[
\frac{\phi^k}{\sqrt{5}}\approx Fibonacci(k)\leq l(\Pi_k)
\]
where $\phi=1.618$

From definitions in section~\ref{terminology} the set $Fib$ of all trees that correspond with the proofs of $A_1\imply yA_n$ is a \SP\; set of trees. A fascinating phenomenon, that happens with this set of huge proofs, concerns the main result of this article. Almost all of these proofs have levels where there are exponentially (or more) many repetitions of a formula, labelling a node of the underlying tree, and these formulas are conclusions of sub-trees that occurs exponentially (or more) many times in these underlying trees.

We have the following proposition on these derivations. 

\begin{proposition}\label{prop:fibonacci}
  Let $\Pi_n$ be the derivation of $A_n$, we drop-off the $\imply$-Intro last rule, from $\eta$ and $\sigma_k$, $k=3,n$. We have thus that
  $occ_{Fib}^{l,A_{n-l}}=Fibonacci(l+1)$, for $l=0,n-1$
\end{proposition}

\begin{proof} of proposition~\ref{prop:fibonacci}
  
  By induction on $n$:
  
  (Basis) $n=1$, $\Pi_1=A_1$ and $occ_{Fib}^{o,A_1}(\Pi_1)=1=Fibonacci(1)$ and for $n=2$ we have that
  \begin{prooftree}
    \AxiomC{$A_1$}
    \AxiomC{$A_1\imply A_2$}
    \BinaryInfC{$A_2$}
  \end{prooftree}
  and hence $occ_{Fib}^{0,A_2}(\Pi_2)=1=Fibonacci(1)$ and $occ_{Fib}^{1,A_1}(\Pi_2)=1=Fibonacci(2)$
\\
  (Inductive step) Let $n>2$ and $0\leq l<n-1$. $\Pi_n$ is the following derivation:
\begin{prooftree}
\def\defaultHypSeparation{\hskip .05in}
\def\ScoreOverang{1pt}
\AxiomC{$\Pi_{n-1}$}
\noLine
\UnaryInfC{$A_{n-1}$}
\AxiomC{$\Pi_{n-2}$}
\noLine
\UnaryInfC{$A_{n-2}$}
\AxiomC{$A_{n-2}\imply (A_{n-1}\imply A_n)$}
\BinaryInfC{$A_{n-1}\imply A_n$}
\BinaryInfC{$A_n$}
\end{prooftree}
For $l=0$, $l=1$ and $l=2$ it is straightforward, consider $l>2$ then
\[
occ_{Fib}^{l,A_{n-l}}(\Pi_n)=occ_{Fib}^{l-1,A_{n-1-(l-1)}}(\Pi_{n-1}) + occ_{Fib}^{l-2,A_{n-2-(l-2)}}(\Pi_{n-2})
      \]
      what is, by inductive hypothesis:
      \[
      occ_{Fib}^{l,A_{n-l}}(\Pi_n)= Fibonacci((l-1)+1) + Fibonacci((l-2)+1) = Fibonacci(l)+Fibonacci(l-1)=Fibonacci(l+1)
        \]

\end{proof}

We also can prove by induction that $h(\Pi_n)=n$ and that $B(\Pi_n)=4\times(n-1)$, for $n>2$. Thus this set $Fib$ of proofs is a set of huge proofs, linear height bounded. By the proposition~\ref{prop:fibonacci} for almost all derivations in $Fib$, there is at least one sub-derivation that occurs exponentially (or more) many times in it. To see that, observe that almost all $\Pi_n$, which are of exponential size on $n$, in fact,  are of exponential-size on $\size{B(\Pi_n)}=4\times(n-1)$ too. Of course, all of $\Pi_n$ are normal derivations in $\mil$. 

\subsection{Proofs in $\mil$ that a non-Hamiltonian graph is not Hamiltonian}

In the appendix~\ref{appendix:hamiltonian} we remember that a well-known propositional coding of the hamiltonianicity of graphs can be used, in purely implicational propositional logic, to have a Natural Deduction proof for the non-hamiltonianicity of graphs, by negating the previous statement. For any non-hamiltonian graph $G$ the sentence $\neg\alpha_G$ is a certificate for its non-hamiltonianicity. The set of all derivations $\alpha_G$ for $G$ non-hamiltonian is a set of huge proofs. These derivations are linearly height bounded in $\mil$ as we demonstrate in the appendix. By counting formulas instead of symbols, if a simple non-hamiltonian graph $G$ has $n$ nodes, then the normal proof of $\alpha_G$ showing that it has no Hamiltonian cycle has $n^n$ formulas, in the worst case. Of course, the size of the normal proof depends on the topology of the graph. For the Petersen graph, for example, the size is approximately $10^{10}$. Petersen graph has ten nodes. We call this set of huge normal proofs $NHam$, for Non-hamiltonian.

$NHam$ also has almost all with sub-derivations that repeats polynomially many times in each of them. The reason for that relies on the fact that $NHam$ is the set of naive proofs of non-hamiltonianicity. The proofs consider all possible paths. In a graph of $n$ nodes, a (possible) path is any sequence of $n$ positions. The proof is an upside-down decision tree that checks the consistency of each possible path for being a correct. Well, we have many repetitions of sub-paths, so the proof reflects this repetition too. We invite the reader to see this in detail in the appendix. We included this material as a matter of completeness and illustration.

\subsection{Normal Proofs that need exponentially many assumptions}

In \cite{exponential} it is shown a family of purely implicational formulas $\xi_n$, such that, each of them needs $2^n$ assumptions of the same formula, we have to discharge at the end of the proof, in their proofs. Almost all of the normal proofs of $\xi_n$ are hence of exponential. All of them are linearly height bounded. This set of huge proofs, denoted by $Exp$ also has the property that almost every proof in it has sub-proofs that are polynomially many times repeated, in a level of the normal proof of $\xi_n$, let us say, in it. Here it is one more concrete example of what we prove in the next section and as the main result of the article.

\section{Huge proofs are redundant}
\label{sec:Main}

In this section, we show the main result of this article, roughly stated as ``Almost every linearly height-bounded huge proof is redundant''. We show that for any unlimited set $\mathcal{S}$ of $EOL$-trees if $\mathcal{S}$ is \SP\; then for almost all trees $T\in \mathcal{S}$, there is a subtree $T^{\prime}$ of $T$ that occurs exponentially (or more) many times at the same level in $T$. The formal statement of this assertion is lemma~\ref{lemma:Main}. The following lemma is an initial step in the proof of lemma~\ref{lemma:Main}. We need one more auxiliary definition before.

\begin{definition}\label{def:OCC}
  Let  $\mathcal{A}$ be a set of $EOL$-trees and $occ^{l,q}_{\mathcal{A}}$ as in definition~\ref{def:occ}. We define the function\linebreak  $OCC^{l,q}_{\mathcal{A}}:\mathbb{N}\longrightarrow\mathbb{N}$ as:
  \[
  OCC^{l,q}_{\mathcal{A}}(m)=Min(\{occ^{l,q}_{\mathcal{A}}(T):\mbox{$\card{(B(T)}= m\land (occ^{l,q}_{\mathcal{A}}(T)>0)$}\})
  \]
\end{definition}

A straightforward consequence of the definition above is that for every tree $T\in\mathcal{A}$, for every $m\in \mathbb{N}$, for every $l\leq h(T)$, if $\card{B(T)}=m$ then for every $q\in B(T)$, $OCC^{l,q}_{\mathcal{A}}(m)\leq occ^{l,q}_{\mathcal{A}}(T)$.

%\td{Fazer a Definicao abaixo ser uma funcao. A funcao e que tem corpotamento  polynomial}

\begin{observation}\label{obs:polypoints}
We first observe that for all trees, the first levels have few nodes, starting with level 1 that has at most 2 nodes, 2 with at most 4, and so on. That is, for every tree $T$, there is at least one level $i\leq h(T)$ and $0<p\in\mathbb{N}$, such that, for all sets $V^{i,q}(T)\leq \card{B(T)}^p$,   for each  $q\in B(T)$, with $V^{i,q}(T)\neq\emptyset$ obviously. 
\end{observation}

\begin{definition}\label{def:MaxOCC}
  Considering the conditions in definition~\ref{def:OCC}, we define:
  \[
  OCC^{l}_{\mathcal{A}}(m)=Max(\{OCC^{l,q}_{\mathcal{A}}(m):\mbox{$\card{B(T)}=m$ $\land$ $q\in B(T)$}\}
    \]
    \end{definition}

In contrats with the observation above,  huge proofs have levels $i$, such that $V^{i,q}(T)$, for some $q\in B(T)$, is lower-bounded by an exponential or faster growing function. The proof of this resembles the proof of lemma~\ref{lemma:SPPHP}.

\begin{lemma}\label{lemma:OCCSP}
Let  $\mathcal{A}$ be the set of all $EOL$-trees and $\Phi(T)$ the predicate that holds only when $T$ is linearly height-bounded, on $\card{B(T)}$, tree.
    If $\Phi(\mathcal{A})$ is \SP\; then $OCC^{l}_{\Phi(\mathcal{A})}$ is \SP\;, on the argument ($m$ in the definition~\ref{def:OCC}) that bounds $\card{B(T)}$ for almost all $l\in\mathbb{N}$.

  \end{lemma}

\begin{proof} of lemma~\ref{lemma:OCCSP}
  Since $\Phi(\mathcal{A})$ is \SP\; then by proposition~\ref{prop:FuncaoF} $F_{\mathcal{A},\Phi}$ is a exponential or fatster growing function from $\mathbb{N}$ into $\mathbb{N}$. We observe that, for any tree $T$,
\[
\size{V(T)}=\size{\bigcup_{i\leq h(T)}\bigcup_{q\in B(T)} V_{T}^{i,q} }
\],
hence, we have that
\[
\size{V(T)}=\sum_{i\leq h(T)}\size{\bigcup_{q\in B(T)} V_{T}^{i,q} }=\sum_{i\leq h(T)}\sum_{q\in B(T)} \size{V_{T}^{i,q}}
\]
The above equality holds for $T\in\Phi(\mathcal{A})$ and for all $m$, such that $\card{B(T)}= m$. Thus, we have
$\size{V(T)}\geq \size{Min_{\mathcal{S}}(\Phi_{\mathcal{S},m})}$. By the fact that $h(T)$ is linearly bounded by $\card{B(T)}$ we have that $h(T)$ is linearly bounded by $m$. Let $k\times m$ be such bound. Hence, we have that:
\[
F_{\mathcal{A},\Phi}(m)\leq \sum_{i\leq h(T)}\sum_{q\in B(T)} \size{V_{T}^{i,q}} = \sum_{i\leq k(m)}\sum_{q\in B(T)} \size{V_{T}^{i,q}}
\]
for almost all $m$, and  almost all $T$. In the right-hand side (RHS) of the last equation we have to remember
$\card{B(T)}= m$, so the inner summation is linearly bounded by $m$.  Moreover, by hypothesis, $\Phi(\mathcal{A})$ is \SP. Thus, $F_{\mathcal{A},\Phi}$ is \SP~and there are $a\in\mathbb{Q}$, $a>1$, and, $c\in\mathbb{Q}$, $c>0$, such that:
\[
a^{cm}\le \sum_{i\leq k(m)}\sum_{q\in B(T)} \size{V_{T}^{i,q}}
\]
for almost all $m$ and all $T$, such that $\card{B(T)}=m$. The above RHS can be seen as a function from $\Phi(\mathcal{A})\times\mathbb{N}$ into $\mathbb{N}$. If for almost all $m$, for all $l\leq k\times m$ and $q\in B(T)$, with $\card{B(T)}\leq m$ we have that $V_{T}^{i,q}$ is polynomially bounded then we have that the mentioned RHS is polynomially bounded also. Hence, there must be $i\leq k\times m$ and $q\in B(T)$, with $\card{B(T)}=m$, such that $\size{V_{T}^{i,q}}>a^{cm}$. Thus, $occ^{l,q}_{\mathcal{A}}(T)>a^{cm}$, for almost all $m$ and all $T$ with $\card{B(T)}=m$. Then, by definition~\ref{def:OCC}, there is $m_0$, such that for all $m>m_0$, there are $k,l$, $l\leq k\times m$, such that
$a^{cm}\le OCC^{l}_{\Phi(\mathcal{A})}(m)$. So $OCC^{l}_{\Phi(\mathcal{A})}$ is \SP\ too.
\end{proof}

The above lemma is, indeed, stronger than what we need to show that almost all normal proof is redundant.
However, from it we can start our reasoning on the intrinsec redundancy of huge proofs. It states that any huge proof $T$ has a level $i\leq h(T)$ and a label $q\in B(T)$ such $a^{cm}\leq \size{V_T^{i,q}}$, for $m=\card{B(T)}$, and some $a>1$ and $c>0$, $a,c\in\mathbb{Q}$. In fact, there is an $m_0$, such that for all $m>m_0$, the previous statement holds. 

We know that $a^{cm}=2^{cm\log(a)}$
Before we prove the main result of this article,  we need one more lemma. 

\begin{lemma}\label{lemma:premises}
  Let $q$ be a label node in a $EOL$-tree $T=\langle V,E_{L}\cup E_{R}\cup E_{U},l,B\rangle$ derived from a normal proof of some formula $\alpha$, such that $B(T)$ is the set of formulas associated to each subtree (non-leaf) $T_{\alpha}$.  Let
\[
Label^{2}_{suc}(q)=\{(c,d):\mbox{$\left<u,v_1\right>\in E_{L}$, $\left<u,v_2\right>E_{R}$, $l(u)=q$, $l(v_1)=c$ or $l(v_2)=d$}\}
\]
be as above. Then $\size{Label^{2}_{suc}}<\frac{1}{3}\times\card{B(T)}$
\end{lemma}

%\todo[linecolor=red,backgroundcolor=red!25,bordercolor=red]{Tem que mudar a relacaoo dos labels na def de $EOL$-tree}
\begin{proof} of lemma~\ref{lemma:premises}
  By the definition of $EOL$-trees, whenever we have three nodes $u,v_1,v_2$, such that, $\left<u,v_1\right>\in E_L$ and $\left<u,v_2\right>\in E_R$, we have that $l(u)\prec l(v_2)$ and $l(v_1)$ is the unique label such $l(v_1)\odot l(u)=l(v_2)$. Since $T$ is derived from a normal proof, $l(v_2)$ is $l(v_1)\imply l(u)$ and by lemma~\ref{lemma:UniqueBranches} there are at most  $\frac{1}{3}\times\card{B(T)}$ different possible branches from any $u$ the the topmost associated to the major premise of the topmost binary rule in $T$. 
\end{proof}

Using the above lemma we can draw the following result.

\begin{lemma}\label{lemma:BinaryExpo}
  Let $\mathcal{S}$ be a set of $EOL$-trees. Suppose that there are $a\in\mathbb{R}$, $a>1$, $n_0,p\in\mathbb{N}$, $c\in\mathbb{Q}$, $p>1$, $c>0$, such that, $\forall n>n_0$, for all $T\in\mathcal{S}$, If $\size{B(T)}=n$ then there is 
a level $i>0$ in $T$ and a label $q\in B(T)$, such that, $\size{V_T^{i,q}}\ge c\times a^{n^p}$, i.e., $V_T^{i,q}$ has exponentially many elements on $\size{B(T)}$. Thus, there is at least one pair of labels $l_1,l_2\in B(T)$, such that $l_2=l_1\odot q$, and $\size{V_T^{i+1,l_1}}=\size{V_T^{i+1,l_2}}\ge (a^{\frac{a-1}{a}})^{n^p}$. Moreover,
  for every $u_1\in V_T^{i+1,l_1}$ there is only one $u_2\in V_T^{i+1,l_2}$ and only one $u\in V_T^{i,q}$, such that, $\langle u,u_1\rangle\in E_L(T)$ and $\langle u,u_2\rangle\in E_{R}(T)$.
\end{lemma}

\begin{proof} of lemma~\ref{lemma:BinaryExpo}
  By the hypothesis of the lemma,  we have that for all $n>n_0$, and, for every $T\in\mathcal{S}$, such that $\size{B(T)}=n$, there are $i$ and $q$, such that, $\size{V_T^{i,q}}\ge c\times a^{n^p}$, with $a>1$, $c>0$, $p>1$. By lemma~\ref{lemma:premises}, there are at most $\frac{1}{2}\times\card{B(T)}=\frac{1}{2}\times n$ possible different pairs of labels, i.e., the sets $Left=\{v:\mbox{$\langle u, v\rangle\in E_L(T)$ and $u\in V_{T}^{i,q}$}\}$ and  $Right=\{v:\mbox{$\langle u, v\rangle\in E_R(T)$ and $u\in V_{T}^{i,q}$}\}$ have the same cardinality :
  \[
  \frac{\size{V_{T}^{i,q}}}{\frac{n}{2}}
  \]
  This is lower-bounded by $\frac{c\times a^{n^p}}{\frac{n}{2}}$, that is the same of $(2c)\frac{a^{n^p}}{n}$, that is equal to $(2c)\frac{a^{n^p}}{a\log_{a}n}=(2c)a^{(n^p-\log_{a}n)}$. Since, for all $n>a$, we have that $n^a>a^n$, thus
  $\log_{a}n^{a}<\log_{a}a^n=n<n^p$, and hence, $\log_{a}n<\frac{n^p}{a}$. Thus $n^p-\log_{a}n>n^p-\frac{n^p}{a}$, and finally
  \[
  \frac{\size{V_{T}^{i,q}}}{\frac{n}{2}}>(2c)a^{(n^p-\log_{a}n)}>(2c)a^{(n^p-\frac{n^p}{a})}=(2c)a^{\frac{a-1}{a}n^p}=(2c)(\sqrt[a]{a}^{a-1})^{n^p}
  \]
  We observe that as $a>1$ then $\sqrt[a]{a}^{a-1}>1$. 
\end{proof}
  Finally, by the last observation above, we can conclude that if there is an unlimited set of trees that each tree has a level and a label that repeats exponentially often, then in the levels above this levels there are labels that repeat exponentially often too.  We have, hence, the main lemma below. 

\begin{lemma}[Every linearly height bounded huge tree is inherently redundant]\label{lemma:Main}
  Let  $\mathcal{A}$ be the set of all $EOL$-trees and $\Phi(T)$ the predicate that holds only when $T$ is of linear height, on $\card{B(T)}$, bounded tree.
    If $\Phi(\mathcal{A})$ is \SP\; then, for almost all $T$, such that $\Phi(T)$, there is a sub-tree $T^{\prime}$ of $T$, and a level $i\leq h(T)$, such that $occ^{i,l(r(T^{\prime}))}$ is \SP, and there is no level $j<i$ such that $occ^{j,q}$ is \SP; for any $q\in B(T)$.
\end{lemma}

\begin{proof} of lemma~\ref{lemma:Main}.
The essential reasoning in this proof is to use the observation~\ref{obs:polypoints} together with lemma~\ref{lemma:OCCSP} above. The observations are used to prove that the set of levels that has no label repeated exponentially or more many times is not empty. 
  The lemma~\ref{lemma:OCCSP} encapsulate  this reasoning. A first observation is that by lemma~\ref{lemma:OCCSP} for almost all trees $T\in\mathcal{A}$ there is a level $i\leq h(T)$, such that $\sum_{q\in B(T)}V^{i,q}$ is exponential. The last phrase is a kind of abuse of language, but it facilitates the understanding of the argumentation. Formally, we had to state that given an exponential lower-bound $a^{cm}$, we prove that there is $i$, such that,  $a^{cm}\le\sum_{q\in B(T)}V^{i,q}$, from the hypothesis, that $a^{cm}\le V(T)$, but we can use the lemma to infer this existence directly.
  
  Turning back to the proof, since $\sum_{q\in B(T)}V^{i,q}$ is exponential (or longer than), using lemma~\ref{lemma:SPPHP} we can conclude that there is $q\in B(T)$ , such that, $V_{T}^{i,q}$ is exponential or longer than exponential. We choose, $i$ and $q$ among the least possible values. Thus, any level $j$, $j<i$, does have sub-exponentially many nodes labeled by each label $q^{\prime}\in B(T)$.
  Any node $v$ in level $i$, with $l(v)=q$, can belong to one and only one of three disjoint sets $Bin$, $Un$ and $Zero$, according if it has two, only one, or none children in $T$. Since $V_{T}^{i,q}= Bin+Un+Zero$ is exponentially lower-bounded, then by lemma~\ref{lemma:SPPHP} at least one of them is exponentially lower-bounded too. Of course, more than one can be exponentially lower-bounded, but we proceed the proof, without generality loss, for each case separately. We prove by induction of the minimal distance of the elements of $V_{T}^{i,q}$.
  \begin{itemize}
     \item The basis case is when one of the  subsets of $V_{T}^{i,q}$ is the set $Zero$. In this case, the distance is 0, for the elements of $Zero$ are leaves themselves. Then, the sub-tree $T^{\prime}$ is formed by each $v\in Zero$ itself, labelled with $q$. They occur exponentially or more than exponentially many times in $T$.
  \item The subset, provided by lemma~\ref{lemma:SPPHP} is $Un$, hence the set $Un^{-}=\{u:\mbox{$\left<v,u\right>\in E(T)$ and $v\in Un$}\}$ is exponentially or longer too, by the item~\ref{item:unique-U} of definition~\ref{EOLTree}. Thus, by inductive hypothesis,  there is a sub-tree $T^{\prime\prime}$ that has the elements of $Un^{-}$ as roots. Thus, we have exponentially or longer many sub-trees of the form $T^{\prime\prime}$ occurring in $T$, with roots in level $i+1$. By adjoining the nodes in $Un$ to these trees, we obtain exponentially or longer many trees $T^{\prime}$ occurring in $T$. We remind the induction value of $Un^{-}$ is smaller than the induction value of $Un$.
  \item The case when $Bin$ is an exponentially bigger subset of $V_{T}^{i,q}$ is analogous to the above. There are at most $\frac{1}{3}\size{V_{T}^{i,q}}$ right-handed branches $T_{right}$ of formulas labeled by $q$ in level $i$. By an iterated  use of inductive hypothesis in each of the levels of $T_{right}$ we obtain sub-trees $T_{left}$  that occurs exponentially, or more, many times in $T$ in the level $i+1$ on. We join them in trees with root in the level of $i$. We obtain the desired trees occurring with roots in $i$ that are repeated at least exponentially. We have only to observe the use of lemma~\ref{lemma:UniqueBranches} to be sure about the fact that there are only linearly many possible labels for the children of the nodes in $Bin$ and lemma~\ref{lemma:BinaryExpo} does all the job. 
\end{itemize}

\end{proof}

Using the lemma~\ref{lemma:Main} we obtain the following theorem, the main result of this article. Its proof is a simple application of the lemma on $EOL$-trees derived from Natural Deduction proofs, observing that all conditions for being a valid $EOL$-tree are fulfilled by the derived trees from ND of \mil, including the fact that the $\odot$ concrete operation is such that $l_1\odot l_2=l_1\imply l_2$. We only state that a sub-derivation of a Natural Deduction proof/derivation is any full sub-tree of the underlying tree of the Natural Deduction derivation.  

\begin{theorem}
  Every huge proof in $\mil$ is redundant, in the sense that there is a sub-derivation of it that is exponentially or more, many times repeated in some level $i$ of the proof.
\end{theorem}

We have only to remember that the sub-trees may not be sub-proofs, for the discharging of assumptions is below them. One last, but not least, thing to remind the reader is the fact that since the sub-trees are the same, the edges have the  same dependency set . Hence, the bitstring that is used to form the represention, formally,  of Natural Deduction proofs and derivation in as $EOL$-trees do not even need to be altered or manipulated in the demonstrations we made in this article.

\section{Conclusion}\label{sec:Conclusion}

In this article, we proved that redundancy is an inherent property of huge normal proofs in Natural Deduction for $\mil$. The extension of this result to the full language of propositional logic seems to hold. The proof of this result, however, can be more complicated and of no evident use. About an application of the result reported here. We can apply it directly to have a compression algorithm to propositional $\mil$ proofs. Using the redundancy of normal proofs, we can improve a previous technique applied in \cite{Studia}, that collapses all nodes labelled with the same label occurring in a given level in a unique node. Instead of collapse only one node, we can collapse a  whole sub-derivations that repeats polynomially many in a super-polynomial proof of a $\mil$-tautology. This new way of collapse produces a simpler version of the horizontal collapse than those that appear in \cite{Studia} to produce arbitrary  $NP$-approximations of the validity in $\mil$, that it is a $PSPACE$ problem. We can say that using the collapse of whole sub-derivations, we obtain  dag-proof for $\mil$-tautologies that are much smaller than their original tree-like forms.
We have an algorithm that performs the verification that our dag-proof is a proof (or not) in cubic time on the size of the dag. 
It is out of the scope of this article to show this algorithm of verification and its time complexity.  

Finally, we have to mention that the redundancy theorem here proved is not a constructible result. Finding the level of the exponential sub-tree is not a feasible task, neither is to find the lowest, in the tree,  possible polynomial level. As further work, we want to use the technique of collapsing sub-trees in the set $NHam$ (see section~\ref{sec:HugeProofs}) to try an alternative\footnote{So to say, a proof that is not a direct consequence of a (possible) proof of $NP=PSPACE$} proof of $CoNP=NP$. The compression method to obtain the certificates for the non-hamiltonicity seems to be unfeasible. This feature has nothing to do with the polynomial-size and polynomial-time verification of any certificate obtained by its means. The feasibility (or not)  of the compression method, so to say the existence of a polytime algorithm to obtain a poly-size and polytime verifiable certificate for minimal tautologies,  has more to do with the conjecture $NP=P$ or $P=CoNP$. 

\section{Appendix}\label{appendix:hamiltonian}

In this appendix we show how to use $\mil$ to generate certificates of the non-hamiltonicity, non-existence of a Hamiltonian cycle, by means of proofs.

A simple directed graph is a directed graph having no multiple edges, that is, for every pair of nodes $(v_1,v_2)$ from the graph there is at most one edge from $v_1$ to $v_2$. 
Given a simple directed graph  $G=\langle V,E\rangle$, with $card(V)=n$, a hamiltonian path in $G$ is a sequence of nodes $v_1v_2\ldots v_n$, such that $v_i\in V$, $i=1,n$, and for each $i,j\in V$, if  $v_i=v_j$ then $i=j$, i.e, a path has no repetition of nodes. Moreover,  if  $v_iv_{i+1}$ is in the path then exists the edge $(v_i,v_{i+1}\in E$. The (decision) problem of knowing whether there is or not a hamiltonian path in a graph is known to be NP-complete. Thus, given a graph $\langle V,E\rangle$, with $card(V)=n$, to verify that a sequence of $n$ nodes is a hamiltonian path it is enough to verify that: (1) There is no repeated node in the sequence; (2) No element $v\in V$ is out of the sequence, and; (3) For each pair $v_i v_j$ in the sequence there is a edge $(v_i,v_j)\in E$. We can see that these verifications require polynomial time on the size of the graph and that any path is linearly bounded by $n$. Thus any sequence of nodes of a graph can be viewed as a polynomially verified certificate for this graph hamiltonicity. Consider the following reduction of Hamiltonian path to SAT, quite well-known from the literature in computational complexity (see \cite{arora}).

\begin{definition}\label{alpha}
Given $G=\langle V,E\rangle$, $card(V)=n>0$. Let $X_{i,v}$, $i=1,n$ e $v\in V$ be the proposicional language. Intuitively, $X_{i,v}$ express that the vertex $v$ is visited in the step $i$ in a path on $G$. Consider the formulas in the following definition:
\begin{enumerate}
\item\label{A} $A=\bigwedge_{v\in V} (X_{1,v}\lor\ldots\lor X_{n,v})$ indicating that every vertex can be visited in any step in a hamiltonian path/cycle;
\item\label{B} $B=\bigwedge_{v\in V}\bigwedge_{i\neq j}(\neg (X_{i,v}\land X_{j,v}))$ indicanting that there are no repetitions in any hamiltonian path/cycle; 
\item\label{C} $C=\bigwedge_{i=1,n}\bigvee_{v\in V} X_{i,v}$ that says that in each step  one vertex should be visited; 
\item\label{D} $D=\bigwedge_{v\neq w}\bigwedge_{i=1,n} \neg(X_{i,v}\land X_{i,w})$ that indicates  that each step can visit at most one vertex, and;
\item\label{E} $E=\bigwedge_{(v,w)\not\in E}\bigwedge_{i=1,n-1}(X_{i,v}\imply \neg X_{i+1,w})$ that indicates that if there is no edge from $v$ to $w$ then $w$ cannot be visited immediately after $v$; 
\end{enumerate}
\end{definition}
We can see that $G$ has a hamiltonian path if and only if $\alpha_{G}=A\wedge B\wedge C\wedge D\wedge E$ is satisfiable. Any hamiltonian path $v_1\ldots v_n$ induces a truth-assignment $T$, such that $T(X_{i,w})=true$ if and only if $w=v_i$, that satisfies $\alpha_G$. Conversely, any truth-assignment that satisfies $\alpha_G$ induces a hamiltonian path in $G$. If we denote $SAT_{Cla}$ the set of satisfiable formulas for the classical propositional logic and as $TAUT_{Int}$ the set of tautologies for the intuitionistic propositional logic, we can observe that the following statements are equivalent:
\[
(1) \mbox{$G$ is not hamiltonian if and only if $\alpha_{G}\not\in SAT_{Cla}$}
\]
\\
\[
(2) \mbox{$G$ is not hamiltonian if and only if $\alpha_{G}$ is unsatisfiable}
\]
\\
\[
(3) \mbox{$G$ is not hamiltoniano if and only if $\neg\alpha_{G}\in TAUT_{Cla}$}
\]
\\
\[ 
(4) \mbox{$G$ is not hamiltonian if and only if $\neg\alpha_{G}\in TAUT_{Int}$}
\]
Hence, $G$ is non-hamiltonian graph if and only if there is an intuitionistic proof (positive certificate) for $\neg\alpha_{G}$.
Such proof is a certificate for non-hamiltonicity of graph $G$.
To go from statement (3) to (4), we use Glyvenko theorem.
In \cite{Statman79,Haeusler2014} it is described a translation from formulas in the full language $\{\bot,\neg,\land,\lor,\imply\}$ to the purely implicational formulas, .i.e, formulas containing only the constant symbol $\imply$ and propositional variables. From any formula $\gamma$, the formula $\gamma^{\star}$ from purely implicational minimal logic is provable in the minimal logic if and only if $\gamma$ is provable in intuitionistic logic. Moreover, concerning the sizes of the formulas, we have that $size(\alpha^{\star})\leq size^{3}(\alpha)$ (\cite{Haeusler2014}). The main idea described in \cite{Statman79,Haeusler2014} is the use of implicational schemata that simulate the introduction and elimination of Natural Deduction rules. This simulation employs the use of new/fresh propositional variables. For example, for each pair of formulas $A$ and $B$, we add the propositional variable $q_{A\lor B}$ and the formulas
$A\imply q_{A\lor B}$, $B\imply q_{A\lor B}$ are used to simulate the $\lor$-introduction rules.
Hence, any application of the rule:
\begin{prooftree}
  \AxiomC{$A$}
  \UnaryInfC{$A\lor B$}
\end{prooftree}
is replaced by the following derivation in $ND_{\imply}$
\begin{prooftree}
  \AxiomC{$A$}
  \AxiomC{$A\imply q_{A\lor B}$}
  \BinaryInfC{$q_{A\lor B}$}
\end{prooftree}
In this way the new derivation is normal too. Remember that the changing of language, i.e., replacing the formula $A\lor B$ by $q_{A\lor B}$ is performed in all formulas of the original derivation. 
The formulas $(A\imply \beta)\imply((B\imply\beta)\imply(q_{A\lor B}\imply\beta))$, for each $\beta$ sub-formula from $\neg\alpha$, simulate the $\lor$-elimination. The fact that the original derivation is normal ensures that any application of an $\lor$-elimination has minor premisses as sub-formulas of the hypotheses or of the conclusion, for the sub-formula principle holds for normal derivation. However,the translation from the intuitionistic full language to the purely implicational language is required to translate any intuitionistic tautology to its implicational tautological counterpart. However, when translating certificates of non-hamiltonicity we can be more economical, as we explain below.
 
A (normal) proof of $\neg\alpha_{G}$, where $G$ is a non-hamiltonian graph, is a proof of $\bot$ from $\alpha_{G}$. Since $\alpha_{G}$ is a conjunction we can consider the certificate of non-hamiltonicity as any proof of $\bot$ from the set of formulas that form the components of the conjunctions. Thus, we consider (A) the disjunctions of the form $(X_{1,v}\lor\ldots\lor X_{n,v})$, with $v\in V$, from the item~\ref{A} of the definition of $\alpha_{G}$; (B) the formulas $\neg (X_{i,v}\land X_{j,v})$, with $i=1,n$ and $v\in V$, from item~\ref{B}; (C) the formulas $\bigvee_{v\in V} X_{i,v}$, with $i=1,n$, from item~\ref{C}; (D) the formulas $\neg(X_{i,v}\land X_{i,w})$, $i=1,n$ and $w\in V$, from item~\ref{D}, and; (E) finally, the formulas $X_{i,v}\imply \neg X_{i+1,w}$, with $i=1,n-1$ and $(v,w)\not\in E$, from item~\ref{E}. 
Let us examine a bit more the proof of $\neg\alpha^{\star}$ as a proof of $\bot$ from this set of formulas $S_{G}$ just detailed above.
A naive proof of non-hamiltonicity proceeds by considering every possible path starting at every possible node of the graph, and then, considering every possible node of the graph to be visited at the second step and so on. Of course, as $G$ has no hamiltonian path, every possible choice ends up in a contradiction. Either there is no unvisited node possible to be visited and we use formulas from the conjunctions in $B$ and  $D$ to obtain the absurdity $\bot$ or there is no node to visit at all, and we use conjunctions from $E$ to prove the $\bot$. Each step's choice is accomplished by a $\lor$ elimination having a disjuntion from the formula $A$ and/or $C$ as major premisse. The proof is a kind of decision tree upwards-down. It is size is upper-bounded by $n^n$, while its height is linear on $n$. We observe that each $\lor$-elimination is replaced by a combination of $\imply$-Introductions and $\imply$-eliminations that increase the height of the tree by 2 for each application of $\lor$-elimination. See schematic proof below. There is no need of $\lor$-introduction. There is no need of the $\bot$ intuitionistic rule. In fact the translation in \cite{Statman79} that replaces the $\bot$ by a new/fresh proposicional variable $q$ is used formally with the sake of uniformity.

Summing up, for each non-hamiltonian simple and directed graph $G$, with $n$ nodes, there is a set of formulas $S_{G}$ with complexity (lenght) at most $n^3$ and a proof of size $n^n$ rules and height at most $3n$ rules that has $q$ as conclusion. This proof is the certificate of non-hamiltonicity of $G$. This proof is also a normal proof, as defined by Gentzen and Prawitz. Hence, it satsifies the sub-formula principle. 
By means of the result explained afterwards in this article we can show that there is a polynomial certificate, polynomially verified in time,  for the non-hamiltonicity of $G$. Since non-hamiltonicity of simple and directed graphs is CoNP-complete then we conclude that NP=CoNP.   

\subsection{Normal $ND_{\imply}$  proofs of non-hamiltonicity of simple directed graphs}

Consider a non-hamiltonian graph $G=\langle V,E\rangle$ and the formula $\neg\alpha_{G}$ as stated in the previous section. Thus, $\alpha_{G}=A\wedge B\wedge C\wedge D\wedge E$, with $A$, $B$, $C$, $D$ and $E$ as defined in items~\ref{A} until~\ref{E} from definition~\ref{alpha}. Let $p$ be any sequence of nodes from $V$ of length  $card(V)$. From $p$ we have the set $\{X_{1,p[1]},\ldots,X_{n,p[n]}\}$ of propositional variables from the language of $\alpha_{G}$. This set, as the sequence $p$, represents a potential path in the graph $G$, namely, the path that starts by visiting vertex $p[1]$, this is  $X_{1,p[1]}$ holds, visits vertex $p[2]$, i.e. $X_{2,p[2]}$ holds, until ending with the visit of $p[n]$. However, this sequence is not checked as a valid path. The set $\{X_{1,p[1]},\ldots,X_{n,p[n]}\}$ is a valid path if and only if it does not inconsistent with the formula $\alpha_G$. As $G$ does not have any hamiltonian path then we known that $\{X_{1,p[1]},\ldots,X_{n,p[n]}\}$ is inconsistent with $\alpha$. We consider a mapping that given a sequence $p$ the set $X_p=\{X_{1,p[1]},\ldots,X_{n,p[n]}\}$ is inconsistent with $alpha_G$, $p\mapsto X_p$. Thus, for any $p$, there is a derivation of $\bot$ from $\alpha$ in Natural Deduction by completeness. Using the translations described in \cite{Statman79, Haeusler2014} we have a normal derivation of $q$ from $alpha_{G}^{\star}$ and the set $X_p$. This is the content of lemma~\ref{paths}. Considering the set $P$ of all sequences of lenght $n$ and the lemma~\ref{paths} we have a set of normal proofs $\{\Pi_p : p\in P\}$, where each $\Pi_p$ is of the form:
\begin{prooftree}
  \AxiomC{$X_p$}
  \noLine
  \UnaryInfC{$\Pi_p$}
  \noLine
  \UnaryInfC{$q$}
\end{prooftree}
where sometimes we use $\Pi_p$ to denote the above derivation. Moreover, in order to have an easier understanding, the derivation $\Pi_p$ is taken as depending from the whole set $X_p$ even when not every formula $X_{i,p[i]}$, $i=1,n$, occurs in $\Pi_p$. We consider the set of nodes (vertexes) ordered as in $\{v_1,\ldots,v_k\}$, where $k=card(V)$. We introduce the following notations:
\begin{definition}
  Given a sequence of vertexes $p:\{1,\ldots,n\}\mapsto V$, we denote the sub-sequence $p[1]\ldots p[j]$ of $p$ as $p[1..j]$, where $j\in\{1,\ldots,n\}$. Of course, $p[1]=p[1..1]$ and $p[1..n]=p$. Moreover, we denote by $p[-j]$ the sub-sequence $p[1..n-j]$, $j\leq n-1$. Obviously, $p[-(n-1)]=p[1]$. The concatenation of sequences $p$ and $q$ is denoted by $p;q$. 
\end{definition}
In the sequel, given a set $X_p=\{X_{1,p[1]},\ldots,X_{n,p[n]}\}$, we denote by $X_{p[-1]}$ the set $\{X_{1,p[1]},\ldots,X_{n-1,p[n-1]}\}$. When dealing with sets of formulas $A$, the union $A\cup\{\beta\}$ can be denoted by $A,\beta$.

Consider now the following derivations $\Pi_{p[-1]}$, for each sequence $p\in P$, that use the set of normal proofs $\{\Pi_p:p\in P\}$ defined as above.
\begin{center}
  $\Pi_{p[-1]}$=\begin{prooftree}
  \AxiomC{$ X_{n,v_1}\lor\ldots\lor X_{n,v_n}$}
  \AxiomC{$X_{p[-1],X_{n,v_1}}$}
  \noLine
  \UnaryInfC{$\Pi_{p[-1];X_{n,v_1}}$}
  \noLine
  \UnaryInfC{$q$}
  \AxiomC{$\ldots$}
  \AxiomC{$X_{p[-1],X_{n,v_n}}$}
  \noLine
  \UnaryInfC{$\Pi_{p[-1];X_{n,v_n}}$}
  \noLine
  \UnaryInfC{$q$}
  \QuaternaryInfC{$q$}
  \end{prooftree}
  \end{center}
If we consider that each $\Pi_p$ is a derivation in $ND_{\imply}$, cf. lemma~\ref{paths},  then the following derivation is the derivation $\Pi_{p[-1]}$ translated to purely implicational minimal logic Natural Deduction, i.e. $ND_{\imply}$. We used above and from now on a n-ary version of the $\lor$-elimination rule, with the sake of a shorter presentation. 
{\small
\begin{center}
  $\Pi_{p[-1]}$=\begin{prooftree}
  \AxiomC{$ORX_{n}$}
  \AxiomC{$X_{p[-1],[X_{n,v_2}]}$}
  \noLine
  \UnaryInfC{$\Pi_{p[-1];X_{n,v_2}}$}
  \noLine
  \UnaryInfC{$q$}
  \UnaryInfC{$X_{n,v_2}\imply q$}
  \AxiomC{$X_{p[-1],[X_{n,v_1}]}$}
  \noLine
  \UnaryInfC{$\Pi_{p[-1];X_{n,v_1}}$}
  \noLine
  \UnaryInfC{$q$}
  \UnaryInfC{$X_{n,v_1}\imply q$}
  \AxiomC{$(X_{n,v_1}\imply q)\imply ((((X_{n,v_2}\imply q)\imply\ldots((X_{n,v_n}\imply q)\imply (ORX_{n}\imply q)))$}
  \BinaryInfC{$(\ldots(X_{n,v_2}\imply q)\ldots ((X_{n,v_n}\imply q)\imply (ORX_{n}\imply q)))$}
  \BinaryInfC{$(\ldots(X_{n,v_3}\imply q)\ldots ((X_{n,v_n}\imply q)\imply (ORX_{n}\imply q)))$}
  \noLine
  \UnaryInfC{$\vdots$}
  \UnaryInfC{$ORX_{n}\imply q$}
  \BinaryInfC{$q$}
  \end{prooftree}
\end{center}
}
The propositional variable $ORX_{n}$ is the translation os the disjunction $X_{n,v_1}\lor\ldots\lor X_{n,v_n}$, as indicated by the translations schemata described in \cite{Statman79} and \cite{Haeusler2014}. Moreover, we can see that the derivations $\Pi_{p[-1]}$, for each $p\in P$ are normal derivations in $ND_{\imply}$.
Since we are building the proof of non-hamiltonicity from the last step to the first, the $j$-th element in the sequence $p$ is related to the choice of visiting vertexes of the graph during step $n-j+1$. The following derivation is, in analogy with previous derivation, regarded to this $j$-th choice in the sequence, and is denoted by $\Pi_{p[-j]}$. Use use $k=n-j+1$ to have a cleaner derivation. Observe that the propositional variables $ORX_{n}$ to $ORX_{k}$ were introduced as the translations of the corresponding disjunctions $X_{n,v_1}\lor\ldots\lor X_{n,v_n}$ to $X_{k,v_1}\lor\ldots\lor X_{k,v_n}$.
\vspace{1cm}

$\Pi_{p[-j]}$=

{\tiny
\begin{prooftree}
  \AxiomC{$ORX_{k}$}
  \AxiomC{$ORX_{n},\ldots,ORX_{k+1},X_{p[-j],[X_{k,v_2}]}$}
  \noLine
  \UnaryInfC{$\Pi_{p[-j];X_{k,v_2}}$}
  \noLine
  \UnaryInfC{$q$}
  \UnaryInfC{$X_{k,v_2}\imply q$}
  \AxiomC{$X_{p[-j],[X_{k,v_1}]}$}
  \noLine
  \UnaryInfC{$\Pi_{p[-j];X_{k,v_1}}$}
  \noLine
  \UnaryInfC{$q$}
  \UnaryInfC{$X_{k,v_1}\imply q$}
  \AxiomC{$(X_{k,v_1}\imply q)\imply ((((X_{k,v_2}\imply q)\imply\ldots((X_{k,v_n}\imply q)\imply (ORX_{k}\imply q)))$}
  \BinaryInfC{$(\ldots(X_{k,v_2}\imply q)\ldots ((X_{k,v_n}\imply q)\imply (ORX_{k}\imply q)))$}
  \BinaryInfC{$(\ldots(X_{k,v_3}\imply q)\ldots ((X_{k,v_n}\imply q)\imply (ORX_{k}\imply q)))$}
  \noLine
  \UnaryInfC{$\vdots$}
  \UnaryInfC{$ORX_{k}\imply q$}
  \BinaryInfC{$q$}
\end{prooftree}
}
The derivation that is our goal is the application of the schemata above to the first step. We can see that there are only $n$ $p[-(n-1)]$ sequences, and by the recursive definitions we used above, there are only $n$ derivations $\Pi_{p[-(n-1)]}$ normal derivations of q. Finally, by a last application of the implicational schema for the $\lor$-elimination we obtain a proof of $q$, the propositional variable used a s the translation of the $\bot$, from $ORX_{1}$ to $ORX_{n}$ that are in fact parts of $A$ and from the formulas $D$, $E$ and $B$ used by lemma~\ref{paths} in producing the derivations $\Pi_{p}$, for each $p\in P$. Now, by an iterated application of a series of $\imply$-introduction rules, we obtain a proof the translation of the formula $\neg\alpha$ in the purely implicational minimal logic Natural Deduction.

\begin{lemma}\label{paths}
  Let $G=\langle V, E\rangle$ be a simple and directed non-hamiltonian graph and 
  $p=p[1]\ldots p[n]$ be a sequence of vertexes from $V$. Then there is a (not necessarily unique) derivation
  \begin{prooftree}
  \AxiomC{$X_p$}
  \noLine
  \UnaryInfC{$\Pi_p$}
  \noLine
  \UnaryInfC{$q$}
\end{prooftree}
  where $X_p=\{X_{1,p[1]},\ldots,X_{n,p[n]}\}$ in $ND_{\imply}$.
\end{lemma}
The above derivation mentioned in the above lemma can use, and must use at least one of, the formulas $B$, $C$, $D$ and $E$ from items~\ref{B},\ref{C} and ~\ref{D} from definition~\ref{alpha}

Proof of the lemma: Since the graph $G$ is not hamiltonian then any sequence of $n$ vertexes cannot be a valid path. So using one of the formulas $B$ to $D$ , in their purely implicational form, we derive $q$ (the translation of the absurdity logical constant). Since $p$ is not a valid path on $G$, then at least one of the items must hold abut $p$:
\begin{description}
\item[Visiting a vertex more than once]\label{repeting-vertex} There are $i,j$, $1\leq i<j\leq n$, such that $p[i]=p[j]=v\in V$. In this case consider $i_1$ and $i_2$ the least pair, $i_1<i_2$, such that $p[i_1]=p[i_2]$. $\Pi_p$ is the following derivation:
  \begin{prooftree}
    \AxiomC{$X_{i_2,v}$}
    \AxiomC{$X_{i_1,v}$}
    \AxiomC{$X_{i_1,v}\imply (X_{i_2,v}\imply q)$}
    \BinaryInfC{$X_{i_2,v}\imply q$}
    \BinaryInfC{$q$}
  \end{prooftree}
  where $X_{i_1,v}\imply (X_{i_2,v}\imply q)$ is the translation of $(\neg (X_{i_1,v}\land X_{i_2,v})$ to purely implicational minimal logic. Observe that $\{X_{i_1,v},X_{i_2,v}\}\subset X_p$ and $(\neg (X_{i_1,v}\land X_{i_2,v})$ is a component of the conjunction $B$ from $\alpha$.
\item[Visiting a vertex without any linking edge]\label{invalid-move} There is $i$, $1\leq i< n$, such that $p[i]=y\in V$, $p[i+1]=z\in V$ and there is no $(y,z)\in E$. In this case consider $j$ the least number $1\leq j< n$ such that $p[j]=y$, $p[j+1]=z$ and $(y,z)\not\in E$. $\Pi_p$ is the following derivation:
  \begin{prooftree}
    \AxiomC{$X_{j+1,z}$}    
    \AxiomC{$X_{j,y}$}
    \AxiomC{$X_{j,y}\imply (X_{j+1,z}\imply q)$}
    \BinaryInfC{$X_{j+1,z}\imply q$}
    \BinaryInfC{$q$}
  \end{prooftree}
  where $X_{j,y}\imply (X_{j+1,z}\imply q)$ is the translation of $X_{j,y}\imply \neg X_{j+1,z}$ to purely implicational minimal logic. Observe that $\{X_{j,y},X_{j+1,z}\}\subset X_p$ and $X_{j,y}\imply \neg X_{j+1,z}$ is a component of the conjunction $E$ from $\alpha$.
\end{description}
Since the two items above are the only two possible reasons for a sequence $p$, with lengh $n=card(V)$, not being a valid path, we are done.

\end{document}